\documentclass{article}
\usepackage{fullpage}
\usepackage{graphicx} %
\usepackage{amsmath}
\usepackage{amsthm}
\usepackage{amsfonts}
\usepackage{tikz}
\usetikzlibrary{positioning}

\usepackage{hyperref}
\usepackage{cleveref}
\usepackage[noend]{algpseudocode}
\usepackage{algorithm}

\usepackage[normalem]{ulem}   %

\newtheorem{definition}{Definition}[section]
\newtheorem{theorem}{Theorem}[section]
\newtheorem{corollary}{Corollary}[section]
\newtheorem{lemma}[theorem]{Lemma}

\title{\textbf{Optimal bounds on a tree inference algorithm}}

\author{Jack Gardiner, \\
Lachlan L. H. Andrew, Junhao Gan, Jean Honorio and Seeun William Umboh \\
\emph{School of Computing and Information Systems, University of Melbourne}}

\date{}

\begin{document}

\maketitle
\begin{abstract}
    This paper tightens the best known analysis of Hein's 1989 algorithm to infer the topology of a weighted tree based on the lengths of paths between its leaves.
    It shows that the number of length queries required for a degree-$k$ tree of $n$ leaves is $O(n k \log_k n)$, which is the lower bound.
    It also presents a family of trees for which the performance is asymptotically better, and shows that no such family exists for a competing $O(n k \log_k n)$ algorithm.
\end{abstract}
\section{Introduction}
Tree inference is the study of algorithms that infer a hidden tree based on limited information about the tree. It has important applications in computational biology, graphical models, hierarchical clustering and other areas. This paper considers algorithms using the distance between the leaves of a hidden tree in order to infer its topology. This occurs frequently phylogenetics, where the leaves are species and the hidden tree describes the common ancestors between the species.  However, it has many other applications, such as determining latent variables in graphical models~\cite{choi_learning_nodate} or the topology of power networks~\cite{fang2024three,flynn2023improved,pengwah2021topology}.

The process of evaluating the distance between a pair of nodes is called a \emph{query}, and the performance indicator is the number of queries required.
This may be suitable when the edge weights represent the similarity of high-dimensional data represented by the nodes, such as the similarity between images, or if a physical experiment is required to determine a pairwise similarity.

This paper tightens the best-known analysis of a classic algorithm~\cite{hein_optimal_1989} to show that it has optimal query complexity,  $O(n k \log_k n)$ for the class of trees of $n$ leaves with degree bounded by $k$, and lower complexity for sufficiently unbalanced trees.
This bioinformatics algorithm is 35 years old, but has had somewhat of a renaissance recently, with a quarter of its citations appearing in the last three years, predominantly within the computer science literature.

It is also shown that, for a non-trivial class of unbalanced trees, the query complexity is $o(n k \log_k n)$.  Specifically, it is $O(n C(n))$ where $C(\cdot)$ depends on the degree of imbalance trees, and is constant for sufficiently unbalanced trees.
In contrast, the algorithm of~\cite{brodal_complexity_2001}, which is the best known algorithm with $O(n k\log_k n)$ query complexity, is shown to have complexity $\Omega(n k \log_k n)$.

Lower bounds on the query complexity for any tree inference algorithm have been proven. For leaf-distance queries, \cite{hein_optimal_1989} proved an $\Omega(n^2)$ lower bound if it is known that the hidden tree has $n$ leaves. If the degree of all nodes of the hidden tree is less than some $k$, the lower bound is reduced to $\Theta(n k\log_k(n))$ as proven by \cite{king_complexity_nodate}.

\section{Preliminaries}
\subsection{Problem definition}
In the Tree Inference problem, there is a hidden unrooted tree $T^* = (V,E)$ with $n$ leaves and positive real edge weights $w(x,y)$. The leaves are labelled and non-leaves are unlabelled. The distance $d(x,y)$ between $x$ and $y$ is defined to be the sum of the edge weights on the unique path between them in $T^*$. The algorithm is given the leaves and their labels but not the non-leaves nor the edge weights. The algorithm has access to an oracle which takes two nodes $x$ and $y$ and returns their distance $d(x,y)$.
The goal is to recover the tree $T^*$ and the edge weights $w$ with as few queries as possible. 

Throughout this paper, we assume that $T^*$ does not have degree-2 nodes. This is without loss of generality:   if a path $xyz$ with $y$ of degree~2 is replaced by edge $xz$ whose weight is the sum of the weights of $xy$ and $yz$, then the leaf-to-leaf distances are unchanged.

\subsection{Anchor calculations}
\label{sec:Anchors}
An approach used in some algorithms is to use distance queries to calculate where on the path between leaves $x$ and $y$ a leaf $z$ would be anchored from. The \emph{anchor point} of $z$ on the path from $x$ to $y$ is the node on that path that is closest to $z$.
Note that the graph estimate is constructed incrementally, and the anchor point may not yet be a node of the graph estimate at the time $z$ is being added, and may appear to be part way along an edge.

The exact method used to calculate the anchor point depends on the context, but as an example an intuitive method used by \cite{king_complexity_nodate} is to calculate the following:
\[\sigma(x,y,z) = \dfrac{d(x,y) + d(x,z) - d(y,z)}{2}\]\\
which is the distance from $x$ to the anchor point of $z$, visualized in figure \ref{fig:anch_calc}:\\
\begin{figure}[H] 
\begin{center}
\begin{tikzpicture}
    \node[shape=circle,draw=black] (x) at (-3,0) {$x$};
    \node[shape=circle,draw=black] (y) at (3,0) {$y$};
    \node[shape=circle,draw=black] (z) at (0,-1.5) {$z$};
    \node[shape=circle,draw=black] (a) at (0,0) {$a$};
    
    \node [above left = 3mm of y] (yb) {};
    \node [below left = 0.5mm of y] (yt) {};
    \node [below left = 0.5mm of a] (al) {};
    \node [below right = 0.5mm of a] (ar) {};
    \node [above left = 0.5mm of z, ] (zl) {};
    \node [above right= 0.5mm of z] (zr) {};
    \node [above right = 3mm of x] (xt) {};
    \node [below right = 0.5mm of x] (xb) {};
    
    \path [-] (x) edge node[above=0pt] {$\sigma(x,y,z)$} node [below = 9pt, color = black] {$d(x,z)$} (a);
    \path [-] (y) edge node [below = 9pt, color = black] {$d(y,z)$} (a);
    \path [-] (z) edge node[left] {} (a);
    \draw [color = red] (xb) -- (-0.35,-0.35) -- (zl);
    \draw [color = blue] (zr) -- (0.35,-0.35) -- (yt);
    \path [color = green] (yb) edge node[above, color = black] {$d(x,y)$}(xt);
    
\end{tikzpicture}
\end{center}
\caption{The red, green and blue lines represent $d(x,z)$, $d(x,y)$ and $d(y,z)$ respectively. Since $d(y,z)$ is subtracted from the anchor calculation, it will cancel out sections of the red and green lines, leaving us with $d(x,y) + d(x,z) - d(y,z) = 2d(x,a) = 2\sigma(x,y,z)$. Dividing by 2 gives the equation for this anchor calculation.}
\label{fig:anch_calc}
\end{figure}
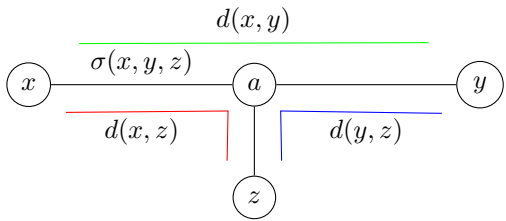

Anchor calculations are used in \cite{king_complexity_nodate} and \cite{hein_optimal_1989}. Anchor calculations are also analogous to triplet queries/experiments used in \cite{brodal_complexity_2001} and \cite{kannan_determining_1996}, as the  between the most closely related pair of leaves in the experiment on 3 leaves $x,y$ and $z$, is the node calculation as the anchor after an anchor calculation.

\subsection{Hein's query algorithm} \label{sec:Hein's algo}
Hein's algorithm~\cite{hein_optimal_1989} is as follows.

At a high level, the algorithm infers the hidden tree by inserting one leaf at a time. It maintains a tree $T$ that initially consists of two leaves $z_1$ and $z_2$ connected by a single edge. It then iterates over the remaining leaves, and for each leaf, inserts it into $T$. After inserting leaves $z_1, \ldots z_n$,  the tree $T$ is exactly the subtree of $T^*$ induced by these leaves.

When inserting a new leaf $z$ into the current tree $T$, anchor calculations are performed between a subset of leaves of $T$ and the new leaf $z$. Each anchor calculation will reduce the number of possible positions $z$ could be attached to. When there is only one position $z$ can be placed, $z$ is added to the tree.

The anchor calculations performed while placing $z$ are chosen so that they minimize the number of queries  required in the worst case. Initially, an edge $uv$ within $T$ is chosen to split $T$ into two rooted subtrees, $T_u$ and $T_v$, rooted at $u$ and $v$ respectively. This edge is referred to as the ``edge-root". A leaf is then chosen from each of the subtrees, leaf $x$ from $T_u$ and leaf $y$ from $T_v$. The anchor calculation $\sigma(x,y,z)$ is then performed, placing $z$ in either $T_v$ or $T_u$, or on the edge joining them.
In the latter case, the new leaf can be placed immediately;
if $T_u$ has lower complexity than $T_v$, then $u$ is replaced by $z$, or vice versa, with ties broken arbitrarily.
To handle the other cases, assume without loss of generality that $z$ is placed in $T_v$.

To place $z$ within $T_v$, anchor calculations are performed between $z$, $x$, and the leaves of $T_v$. Starting at $v$, the aim is to place $z$ in a subtree below $v$. For the $i$'th child node $v_i$ of $v$, the subtree rooted at $v_i$ is denoted by $T_i$. This is illustrated in Figure~\ref{fig:Rooted_tree1}.
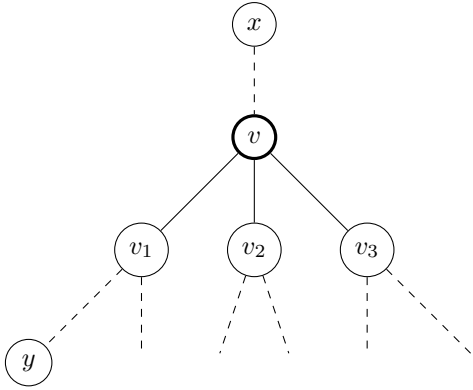
\begin{figure}
\begin{center}
\begin{tikzpicture}
    \node[shape=circle,draw=black] (x) at (0,1.5) {$x$};
    \node[shape=circle,draw=black, very thick] (v) at (0,0) {$v$};
    \node[shape=circle,draw=black] (v1) at (-1.5,-1.5) {$v_1$};
    \node[shape=circle,draw=black] (v2) at (0,-1.5) {$v_2$};
    \node[shape=circle,draw=black] (v3) at (1.5,-1.5) {$v_3$};
    \node[shape=circle,draw=black] (y) at (-3,-3) {$y$};
    \node[] (b1) at (-1.5,-3) {};
    \node[] (b2) at (-0.5,-3) {};
    \node[] (b3) at (0.5,-3) {};
    \node[] (b4) at (1.5,-3) {};
    \node[] (b5) at (3,-3) {};
    
    \path [dashed] (x) edge node[left] {} (v);
    \path [dashed] (v1) edge node[left] {} (b1);
    \path [dashed] (v2) edge node[left] {} (b2);
    \path [dashed] (v2) edge node[left] {} (b3);
    \path [dashed] (v3) edge node[left] {} (b4);
    \path [dashed] (v3) edge node[left] {} (b5);
    \path [dashed] (v1) edge node[left] {} (y);
    \path [-] (v) edge node[left] {} (v1);
    \path [-] (v) edge node[left] {} (v2);
    \path [-] (v) edge node[left] {} (v3);
    
\end{tikzpicture}
\end{center}
\caption{A tree rooted at $v$, with leaf $x$ within $T_u$, and $y$ a leaf in $T_1$}
\label{fig:Rooted_tree1}
\end{figure}
After performing an anchor calculation between $x$, $z$ and a leaf $y$ of $T_1$, there are three possible outcomes. If the anchor point is at $v$, then $z$ must not be placed within $T_1$. If the anchor point is between $v$ and $v_1$, then $z$ is placed on the edge between $v$ and $v_1$. Otherwise, the anchor point must be in $T_1$. No other cases are possible, as it is known that the anchor point is at or below $v$. This process is repeated until a tree $T_i$ is found that $z$ will be placed in, or $a$ is placed on an edge incident to $v$, i.e., for some $i$, we replace the edge $x v_i$ with $x a$ and $a v_i$ and add edge $z a$. If the anchor point is placed at $v$ for each  $T_i$,
then $z$ is inserted as a new child of $v$.

If it is determined that $z$ is to be placed in $T_i$, the process is repeated to place $z$ in a subtree below $v_i$. This is continued until $z$ can be placed directly connected to an edge or vertex, at which point it is inserted into the currently constructed tree.

Once all leaves have been inserted in this manner, the constructed tree will match the hidden tree, and hence the algorithm is complete. \\

Note that each anchor calculation for $z$ after the first requires only one new distance query, $d(y,z)$, as the distance $d(x,y)$ can be inferred from the partial graph without a distance query, and the distance $d(x,z)$ can be re-used.

\subsection{Complexity}
Hein~\cite{hein_optimal_1989} introduced a notion complexity to calculate the maximum number of queries used by Hein's algorithm.
Each node has a complexity, with leaves having complexity 0, and complexities recursively defined moving away from the leaves.
For rooted trees, the complexity is computed recursively toward the root, and complexity of the tree is the complexity of the root, which is the largest complexity of any node.
For unrooted trees, the process is more complicated, and deferred to Section~\ref{sec:unrooted}.

\subsubsection{Complexity of rooted trees}\label{def:Rooted}
The complexity $f(T)$ of a tree $T$ rooted at $r$ is defined
recursively by $f(T) = 0$ if $T$ is a singleton, and otherwise
\begin{equation}\label{eq:complexity_recursion}
    f(T) = \max_{1 \leq i \leq q}{(f(T_i) + i - 1)}
\end{equation}
where $f(T)$ is the complexity of tree $T$, and $T_1,\dots ,T_q$ are the child-trees rooted from the children of $r$ ordered so that $f(T_i)$ is non-increasing.

\subsubsection{Complexity of unrooted trees} \label{sec:unrooted}
For unrooted trees, there is no separation between parents and children, and so the concept of complexity must be modified.
The complexity of an unrooted tree is again calculated recursively from the leaves, towards a centre point which we will call the \emph{edge root}.
The complexity of a node $N$ in an unrooted tree is the complexity of the tree rooted at $N$ formed by removing the edge from $N$ towards the edge root, or the edge root itself if it is adjacent to $N$.
An algorithm for identifying the edge root is given in~\cite{hein_optimal_1989}, in which each iteration identifies tentative complexities for some internal nodes and commits the (equal) smallest of those.

The complexity of an unrooted tree $T$ is defined in terms of functions $f_{uv}$ defined as
\begin{equation}
    f_{uv}(T) = \max \{f(T_u), f(T_v)\} + 2
\end{equation}
where $uv$ is an edge in $T$ and $T_u$ and $T_v$ are the subtrees formed by deleting edge $uv$.
The complexity of $T$ is then
\begin{equation}
    \min_{uv \in E} f_{uv(T)}
\end{equation}
where $E$ is the set of edges of $T$.

\begin{figure}
\begin{center}
\begin{tikzpicture}
    \node[shape=circle,draw=black] (L1) at (0,0) {0};
    \node[shape=circle,draw=black] (L2) at (0,3) {0};
    \node[shape=circle,draw=black] (L3) at (3,0) {0};
    \node[shape=circle,draw=black] (L4) at (1.5,0) {0};
    \node[shape=circle,draw=black] (L5) at (4,-1.5) {0};
    \node[shape=circle,draw=black] (L6) at (5.5,-1.5) {0};
    \node[shape=circle,draw=black] (L7) at (7,-1.5){0};
    \node[shape=circle,draw=black] (L8) at (7,2) {0};
    \node[shape=circle,draw=black] (L9) at (4.5,3.5) {0};
    \node[shape=circle,draw=black] (L10) at (6.5,3.5) {0};
    \node[shape=circle,draw=black] (I1) at (1.5,1.5) {1};
    \node[shape=circle,draw=black, very thick] (I2) at (3,1.5) {2};
    \node[shape=circle,draw=black,very thick] (I3) at (4.25,1.5) {2};
    \node[shape=circle,draw=black] (I4) at (5.5,0) {2};
    \node[shape=circle,draw=black] (I5) at (5.5,2) {1};
    \node[shape=circle,draw=black] (I6) at (5.5,3) {1};
    
    \path [-] (L1) edge node[left] {} (I1);
    \path [-] (L2) edge node[left] {} (I1);
    \path [-] (I1) edge node[left] {} (I2);
    \path [-] (L3) edge node[left] {} (I2);
    \path [-] (L4) edge node[left] {} (I2);
    \path [red] (I2) edge node[left] {} (I3);
    \path [-] (I3) edge node[left] {} (I4);
    \path [-] (L5) edge node[left] {} (I4);
    \path [-] (L6) edge node[left] {} (I4);
    \path [-] (L7) edge node[left] {} (I4);
    \path [-] (I3) edge node[left] {} (I5);
    \path [-] (L8) edge node[left] {} (I5);
    \path [-] (I5) edge node[left] {} (I6);
    \path [-] (L9) edge node[left] {} (I6);
    \path [-] (L10) edge node[left] {} (I6);
\end{tikzpicture}
\end{center} 
\caption{The red edge is the edge chosen to split on, and the nodes in bold are the roots of the rooted trees created by the splits. The complexity of this tree is $\max\{2,2\} + 2 = 4$.
}
\label{fig:Completed_complexity_tree}
\end{figure}
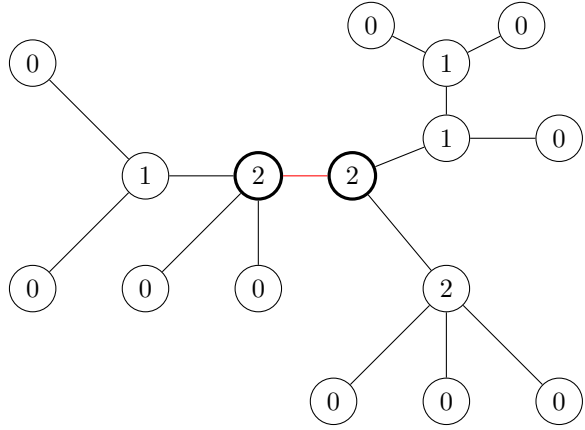

\subsubsection{Asymptotic bounds}
A big-$O$ upper bound on the complexity of a rooted tree with respect to its leaf-count also applies as an upper bound on the unrooted complexity of a tree.
This is because each of the two trees $T_u$ and $T_v$ have less leaf nodes than $T$, and the complexity of $T$ scales linearly with the rooted complexity of $T_u$ and $T_v$. 

In \cite{hein_optimal_1989} the complexity of a rooted tree with $n$ leaves and all leaves with having a degree of at most $k$ was proven to be $O(k\log(n))$. This bound is tightened to $O(k\log_k(n))$ in section \ref{Upper Bound}.
For fixed $k$ this does not change the asymptotic order, but it has an impact if $k$ is allowed to grow.
In particular, if $k = \Theta(n)$ then the complexity reduces from $O(n^2\log(n))$ to the optimal $O(n^2)$.

\section{Upper bound on the query algorithm}\label{Upper Bound}
This section proves that the upper bound for the query complexity of Heins algorithm is $O(nk\log_k(n))$ for trees of bounded degree $k$, whereas Hein proved it to be $O(nk\log(n))$. The bound proven here is tighter, as $n k\log_k(n) = o(n k\log(n))$ when $n \geq k = \omega(1)$. The proof up until Lemma \ref{lem:equal_degree} recapitulates Hein's argumentation in \cite{hein_optimal_1989}.

\subsection{Existing results}
The approach of the proof is to show that an exponential lower bound on how many nodes a tree with a certain complexity can have. This can then be turned into a logarithmic upper bound on the complexity of a tree with $n$ leaves.\\
Let $S(f_0)$ be the minimal number of leaf nodes in a rooted tree of complexity $f_0$, with all nodes having degree less than $k$. When $0 < f_0 < k-1$, $S(f_0)$ is defined as $S(f_0) = f_0+1$, which corresponds to a``star" rooted tree of degree $f_0 + 1 $. 

\begin{theorem}\label{theorem:recursion}
    $S(f_0)$ has the following recursion for any $f_0 \geq k-1$
\[S(f_0) = \min_{1 < i < k}(iS(f_0-i + 1))\]
\end{theorem}
This includes the term ``${}+1$'', which was missing from the corresponding 
expression in~\cite{hein_optimal_1989},
which causes $S(f_0)$ to grow more quickly.
The proof follows the logic of~\cite{hein_optimal_1989}, but is repeated here to justify the additional term.
A rooted tree is called \textit{minimal} if it has complexity $f_0$ and $S(f_0)$ leaf nodes. Theorem~\ref{theorem:recursion} is proven utilizing the following results from~\cite{hein_optimal_1989}.
\begin{lemma}\label{lemma:all_minimal}
    For any minimal tree $T$, all child trees of $T$ are minimal.
\end{lemma}
\begin{lemma}\label{lemma:equal_complexity}
    For any minimal tree $T$, all child trees have the same complexity.
\end{lemma}

\begin{proof}[Proof of Theorem~\protect\ref{theorem:recursion}]
Let $T$ be any minimal tree of complexity $f_0$ and the root node of $T$ have degree~$q$.
By Lemma \ref{lemma:equal_complexity}, all child trees have the same complexity, so $f(T_i) = f(T_q)$ for all possible values of $i$. 
By~\eqref{eq:complexity_recursion},
\[f_0 = \max_{1 \leq i \leq q} {f(T_q) + i - 1} = f(T_q) + q - 1\]
whence $f(T_q) = f_0 - q + 1$.
Since $T$ is minimal, all child trees of $T$ are minimal by Lemma \ref{lemma:all_minimal}. Hence they all have $S(f(T_q))$ leaf nodes, and there are $q$ of them.
The result follows from the minimality of $S(f_0)$.
\end{proof}
From this, \cite{hein_optimal_1989} observed that $S(f_0)$ grows exponentially. As a result the inverse function grows logarithmically, meaning that the number of leaf nodes grows logarithmically with complexity. As shown below, $S(f_0)$ actually grows exponentially with base $k$, leading to a tighter result for large $k$. First, some lemmas about minimal trees must be proven.

The concept of a filled tree and its layer-degree sequence is now introduced. A \emph{filled} tree is a rooted tree where all nodes at a certain unweighted distance from the root have the same degree. The \emph{layer-degree sequence} of a filled tree is a finite sequence $(q_1,q_2,\dots,q_n)$ where $q_i$ is the number of children of each node that is $i-1$ edges away from the root. All of $q_i$ are greater than 1. For an example, the filled tree with layer-degree sequence $(3,2)$ is:\\
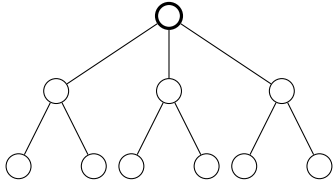
\begin{figure}[H]
\begin{center}
        
    \begin{tikzpicture}
    \node[shape=circle,draw=black, very thick] (I1) at (0,0) {};
    \node[shape=circle,draw=black] (I2) at (-1.5,-1) {};
    \node[shape=circle,draw=black] (I3) at (0,-1) {};
    \node[shape=circle,draw=black] (I4) at (1.5,-1) {};
    \node[shape=circle,draw=black] (I5) at (-2,-2) {};
    \node[shape=circle,draw=black] (I6) at (-1,-2) {};
    \node[shape=circle,draw=black] (I7) at (-0.5,-2) {};
    \node[shape=circle,draw=black] (I8) at (0.5,-2) {};
    \node[shape=circle,draw=black] (I9) at (1,-2) {};
    \node[shape=circle,draw=black] (I10) at (2,-2) {};
    \path [-] (I1) edge node[left] {} (I2);
    \path [-] (I1) edge node[left] {} (I3);
    \path [-] (I1) edge node[left] {} (I4);
    \path [-] (I2) edge node[left] {} (I5);
    \path [-] (I2) edge node[left] {} (I6);
    \path [-] (I3) edge node[left] {} (I7);
    \path [-] (I3) edge node[left] {} (I8);
    \path [-] (I4) edge node[left] {} (I9);
    \path [-] (I4) edge node[left] {} (I10);
\end{tikzpicture}

\end{center}

\caption{A \textit{filled} tree with a layer-degree sequence of $(3,2)$}
\label{fig:filled_tree}
\end{figure}
\noindent Note that the ``degree" in a layer-degree sequence is actually the amount of children a node has. This is for simplicity and brevity of the equations used.

\subsection{Tighter results with $k$-ary trees}
The following lemma can be proven via induction on $f_0$.
\begin{lemma}\label{lem:equal_degree}
    For any complexity $f_0 > 0$, there exists a minimal tree that is also a \textit{filled} tree.
\end{lemma}
Any minimal tree becomes filled by replacing its child-trees with filled minimal trees, whose existence is guaranteed by induction.
The following can also be proved by induction.
\begin{lemma}\label{lemma:filled_complexity}
    Given a filled rooted tree $T$ with a layer-degree sequence of $(q_1,q_2,q_3,\cdots,q_j)$: 
    \begin{itemize}
        \item the complexity of $T$ is $\sum_{i = 1}^j (q_i - 1)$
        \item the number of leaf nodes in $T$ is $\prod_{i = 1}^j q_i$
    \end{itemize}
\end{lemma}

\noindent Neither of these equations relies on the ordering of the layer-degree sequence, allowing the layers to be re-arranged without changing how many leaves are in the tree or its complexity. So given a filled minimal tree, the layer-degrees can be re-arranged and the resulting tree will still be minimal. From this point onwards, all layer-degree sequences are assumed to be in non-increasing order, retaining their minimality if they are minimal.\\

The main result of this section is the following.
\begin{theorem} \label{theorem:upper_bound}
    If a rooted tree $T$ has $n$ leaves and the degree of its nodes are bounded above by $k$, then its complexity $f(T)$ satisfies the following:
    \[f(T) = O(n k\log_k(n)).\]
\end{theorem}
\begin{proof}
Consider a filled minimal rooted tree $T$ with complexity $f_0$ and a layer-degree sequence of $S = (q_1,q_2,q_3,\cdots,q_j)$ in non-increasing order. If the degree bound is $k$, it will be shown that $q_{j-1} = k-1$ by contradiction.

If $q_{j-1} \neq k-1$, then $q_{j-1} < k-1$ because otherwise the degree exceeds the bound. Now, consider the filled tree $T'$ with layer-degree sequence $S' = (q_1,q_2,q_3,\cdots,q_{j-1} + 1, q_jg-1)$. This tree still has the required degree bound as $q_{j-1} < k-1$ and therefore $q_{j-1} + 1 < k$. Furthermore, complexity and leaf count can be calculated as:
\begin{itemize}
    \item The complexity of $T' = $  $\sum_{i = 1}^{j-2} (q_i - 1) + (q_{j-1} + 1 - 1) + q_j - 1 - 1= \sum_{i = 1}^j (q_i - 1)$
    \item The leaf count of $T' = $  $(q_{j-1} + 1)(q_{j} - 1)\prod_{i = 1}^{j-2} q_i = (q_{j-1}q_{j} + q_{j} - q_{j-1} - 1)\prod_{i = 1}^{j-2} q_i$
\end{itemize}
From the non-increasing property of $S$ it follows that $q_j \leq q_{j-1}$. 
Hence the leaf count of $T'$ is
\begin{align*}
    (q_{j-1}q_{j} + q_{j} - q_{j-1} - 1)\prod_{i = 1}^{j-2} q_i
    &< \prod_{i = 1}^{j} q_i
\end{align*}
which is the leaf count of $T$, yet its complexity remains the same. This contradicts the leaf-count minimality of $T$. Therefore $q_{j-1} = k-1$. Furthermore, since $S$ is non-increasing, all degrees above $q_{j-1}$ are at least $q_{j-1}$, therefore
\begin{equation}\label{eq:S}
S = (k-1,k-1,k-1,\cdots, k-1,  q_j).
\end{equation}

Now, $f_0$ can be calculated using Lemma~\ref{lemma:filled_complexity}. Since $f(T)$ is the complexity of $T$,
\[f_0 = f(T) = \sum_{i = 1}^{j-1}(k-2) + q_j - 1 = (j-1)(k-2) + q_j-1\]
whence, since $q_j \leq k-1$,
\[n = \dfrac{f_0-q_j}{k-2} + 1
\geq \dfrac{f_0}{k-2}
\]
By~\eqref{eq:S}, the leaf count $S(f_0)$ of $T$ then satisfies
\begin{equation}\label{eq:S(f_0)}
  S(f_0) = q_j\prod_{i = 1}^{j-1} (k-1) > (k-1)^{{f_0}/(k-2) - 1}
\end{equation}
This shows that $S(f_0)$'s growth is at least exponential, as Hein observed. This formula can now be re-arranged to get an upper bound on the complexity of a tree with $n$ leaves.\\
Given a tree with $n$ leaves, allow it to have complexity $f_0$. Therefore $n \geq S(f_0)$, otherwise $S(f_0)$ is not minimal.
Taking log of~\eqref{eq:S(f_0)} then gives
\begin{align*}
    \log_{k-1}(n) &> \frac{f_0}{k-2} - 1
\end{align*}
whence
\begin{equation}
    f_0 \leq (k-2)\log_{k-1}(n) + k - 2 = O(k\log_k(n))
\end{equation}
Hence any tree $T$ with $n$ leaves and a degree bound of $k$ has $f(T) = f_0 = O(k\log_k(n))$ 
\end{proof}

Now the following theorem can be proven, using Theorem~\ref{theorem:upper_bound}
\begin{theorem}
    The query complexity of Hein's algorithm is $O(nk\log_k(n))$.
\end{theorem}
\begin{proof}
By Theorem~\ref{theorem:upper_bound}, the complexity of any rooted tree with $n$ leaves and a degree bound of $k$ is $O(k\log_k(n))$. The complexity of an unrooted tree is defined as the maximum of the complexities of two rooted trees joined by a chosen ``edge-root" plus 2. Hence, the complexity of an unrooted tree also has a similar bound of $f_0 \leq (k-1)\log_{k-1}(n) + 2$ and therefore $f_0 = O(k\log_k(n))$ for unrooted trees as well.\\
Recall that complexity is the maximum number of queries required to insert a new node into a tree $T$ in Hein's algorithm. If the tree has $n$ leaf nodes, it will require no more than $n$ insertions, and each insertion takes $f(T)$ queries, which has been shown to be bounded above by $O(k\log_k(n))$.
\end{proof}

This matches the lower bound for arbitrary trees with degree bound $k$.
Since \cite{hein_optimal_1989} predates \cite{brodal_complexity_2001}, it appears to be the first algorithm to achieve that bound.

\subsection{Tighter results for unbalanced binary trees}
When the tree is constrained to be sufficiently unbalanced, tighter results are possible.
For ease of exposition, these are presented for binary trees.

The following class of trees has imbalance at each node such that the size of the subtree rooted at one child is smaller than that of the other child, but an amount controlled by a function $g$.
\begin{definition}[$g$-beanstalk]
Let $g:\mathbb{N} \rightarrow [1,\infty)$ be non-decreasing with $g(n) \le n$. An \emph{$g$-beanstalk} is a rooted binary tree such that $N_S(n) \le g(N_L(n))$ for all nodes $n$ where $N_L(n)$ and $N_S(n)$ are respectively the numbers of leaves in the larger and smaller subtrees of the children of $n$.
\end{definition}

\begin{theorem}\label{thrm:beanstalk}
Let $\hat g^{-1}:\mathbb{N} \rightarrow \mathbb{N} \cup \{\infty\}$ with $\hat g^{-1}(n) = \inf\{x \in \mathbb{N} \mid g(x) \geq n\}$, let
\begin{equation}
    h(n) =  \hat g^{-1}(n) + n
\end{equation}
and let $h^{(i)}$ be the $i$-fold composition of $h$, so that $h^{(2)}(n) = h(h(n))$.

The  complexity $C(n)$ of an $g$-beanstalk with $n$ leaves satisfies
\begin{equation}\label{eq:beanstalk_complexity}
C(n) < 
\min\{C \in \mathbb{N} \cup \{0\} \mid n < h^{(C)}(1)\}
\end{equation}
\end{theorem}
\begin{proof}
    Let $N(c)$ be the minimum number of leaf nodes required in a $g$-beanstalk of  complexity $c$.
    Let $T$ be a rooted binary tree of complexity $c$ with $N(c)$leaves and with root $r$, such that neither child subtree has complexity $c$.
    Then $T$ must have two child subtrees both of complexity $c-1$.
    Hence $g(N_L(r)) \ge N_S(r) \ge N(c-1)$,
    whence $N_L(r) \ge \hat g^{-1} (N(c-1))$.
    $T$ has $N_L(r) + N_S(r)$ leaf nodes.
    $T$ was taken to be an $g$-beanstalk of complexity $c$, hence $N(c)$ is subject to
    \begin{equation}\label{eq:N(c)>h(N(c-1))}
    N(c) \geq N_L(r) + N_S(r) \geq \hat g^{-1}(N(c-1)) + N(c-1) = h(N(c-1))
    \end{equation}
    Next, let $T'$ be an $g$-beanstalk of complexity $c$ whose smaller child subtree has complexity $c-1$ and $N(c-1)$ leaves, and whose larger child subtree has complexity $c-1$ and  $\hat g^{-1}(N(c-1))$ leaf nodes.
    This is possible, for example, by making the larger subtree a copy of the smaller subtree with one leaf replaced by a caterpillar of sufficient size.
    Since $T'$ has complexity $c$, it provides an upper bound on $N(c)$ that matches~\eqref{eq:N(c)>h(N(c-1))}, giving the equality\[N(c) = h(N(c-1)) \]
    For any $g$, the singleton tree will have one leaf and a complexity of $0$, and furthermore the singleton tree is binary.
    Hence, $N(0) = 1$ and by induction,
    \[N(c) = h^{(c)}(1)\]
    Consider an $g$-beanstalk with $n$ leaves and some complexity $f_0$.
    Since $N(\cdot)$ is increasing, if $n < N(c)$ for some $c$, then  for all $x \ge c$, it follows that $n < N(c) \le N(x)$ and, by the minimality of $N(x)$, $f_0 \neq x$. Therefore $f_0 < c$.
    Hence
    \[f_0 < \min\{c \in \mathbb{N} | \; n < N(c)\}
    = \min\{c \in \mathbb{N} | \; n < h^{(c)}(1)\}
    \]
    and the result follows by the definition of $C(n)$.
\end{proof}

This admits some simple special cases.
\begin{corollary}
With the notation of Theorem~\ref{thrm:beanstalk}, the following hold.
\begin{enumerate}
\item If $g(n) =  \gamma n$ for some $\gamma \in (0,1]$, then
\[
  C(n) \leq \left\lceil \log_{1+{1}/{\gamma}}(n)\right\rceil.
\]

\item If $n\ge 2$ then a simpler bound is
\[C(n) \leq \min\{C \in \mathbb{N} \mid g^{(C)}(n) < 2\}. \]
Let $\Delta(n) = g(n)-g(n-1)$.  If $\Delta(n)$ is positive decreasing and $\Delta(n) \le 1/\sqrt{n}$, then this relaxation is without loss in the sense that
\begin{equation}\label{eq:tight}
\lim_{n\rightarrow \infty}
    \frac{\min\{C \in \mathbb{N} \mid n < h^{(c)}(1)\}}
    {\min\{C \in \mathbb{N} \mid g^{(C)}(n) < 2\}} = 1.
\end{equation}

\item If $g(n) = n^\gamma$ for some $\gamma \in (0,1)$, then for $n\ge 2$
\[C(n) \leq \lfloor\log_{1/\gamma} \log_2 n \rfloor + 1.
\]

\item If $g(n) = a$, with $a \in \mathbb{N}$, then
\[C(n) \leq \left\lceil \log_{2}(a) +  1\right\rceil.\]
\end{enumerate}
In each case, Hein's algorithm takes $O(n C(n))$ queries to identify the unrooted relaxation of the $g$-beanstalk.
\end{corollary}
\begin{proof}
\begin{enumerate}
\item
If $N(c)$ is the minimum number of leaves for a tree restricted to the conditions of Theorem~\ref{thrm:beanstalk},with complexity $c$, the recursion becomes
\[N(c) = \left\lceil \dfrac{1}{\gamma}N(c-1)\right\rceil + N(c-1)\]
which gives $N(c) \geq (1+1/\gamma)^c$ since $N(0)=1$. %
\item

By~\eqref{eq:N(c)>h(N(c-1))}, $N(c) \ge \hat g^{-1}(N(c-1))$ and so $g(N(c)) \ge g(\hat g^{-1}(N(c-1))) \ge N(c-1)$.
By induction, $g^{(c-1)}(N(c)) \ge N(1) = 2$, and so the smallest $x$ such that $g^{(x)}(N(c)) < 2$ must satisfy $x \ge c$.

Let $n_{i+1} = h(n_i)$, let $C_i^{\text{num}} = \min\{C \in \mathbb{N} \mid n_i < h^{(C)}(1)\}$ and let $C_i = \min\{C \in \mathbb{N} \mid g^{(C)}(n_i) < 2\}$.
Let
\[R_i = \frac{\min\{C \in \mathbb{N} \mid n_i < h^{(C)}(1)\}}
    {\min\{C \in \mathbb{N} \mid g^{(C)}(n_i) < 2\}}
    = \frac{C_i^\text{num}}{C_i}
\]
Each time $i$ increases, the numerator increases by 1, since $n_i < h^{(C)}(1)$ if and only if $h(n_i) < h(h^{(C)}(1)) = h^{C+1}(1)$ since $h$ is non-decreasing.  Hence $C_{i+l}^{\text{num}} = C_i^{\text{num}} + l$.

It remains to show that $C_{i+l} = C_i +l + o(l)$. %

First, we show that
\begin{equation}\label{eq:grow_by_l}
    C_{i+l} \ge C_i+l.
\end{equation}
Now
\begin{equation}
    g^{(C_{i})}(n_{i}) < 2 \le g^{(C_{i}-1)}(n_{i}) \le g^{(C_i)}(h(n_i))
\end{equation}
since $n \le g(h(n))$, and
\begin{equation}
    g^{(C_{i+1})}(h(n_{i})) = g^{(C_{i+1})}(n_{i+1}) < 2
\end{equation}
whence $g^{(C_{i+1})}(h(n_i)) < 2 \le g^{(C_i)}(h(n_i))$ and, since $g(n)<n$, $C_{i+1} > C_{i}$.
Since they are integers, $C_{i+1} \ge C_i + 1$, and~\eqref{eq:grow_by_l} follows.

It remains to show that $C_{i+l} \le C_i + l + o(l)$.
Since $g(n) = o(n)$, for all $l$ there exists an $n_0$ such that for all $n > n_0$,
\begin{equation}\label{eq:concave}
    g(n+l) < n.
\end{equation}
By the definition of $h$ and the assumption $g(n) - g(n-1)$ is non-increasing there exists an $ n_1$ such that for all $n > n_1$,
\begin{align}
    g(h(n)) & = g(\min\{x : g(x) \ge n\} + n\} & \text{by definition} \\
    & = g(\max\{x : g(x) < n\} + 1 + n\} & \text{as both sets are non-empty and $g$ increasing} \\
    & \le g(\max\{x : g(x) < n\}) + 1 & \text{since $\Delta(n)<1/\sqrt{n}$ } \\
    & < n+1.
\end{align}
This implies that for $n > n_1$,
$g^{(i)}(h^{(i)}(n)) < g^{(i-1)}(h^{(i-1)}(n)+1) \le g^{(i-1)}(h^{(i-1)}(n))+1$
and by induction, $g^{(l)}(h^{(l)}(n)) \le n+l$.
Combined with~\eqref{eq:concave}, this shows that, for sufficiently large $n$,
\begin{equation}
    g^{(l+1)}(h^{(l)}(n)) < n
\end{equation}
whence
\begin{equation}
    g^{(C)}(n) < 2 \Rightarrow g^{(C+l+2)}(h^{(C+l)}(n))<2.
\end{equation}
That implies that when $n_i$ is sufficiently large and $i$ increases by $l$  (and hence $C^{\text{num}}_i$ does too), $C_i$ increases by no more than $l+2$.

In conjunction with~\eqref{eq:grow_by_l}, this shows that for sufficiently large $i$, for all $l > 1$,
\begin{equation}
    \frac{C_i^{\text{num}} + l}{C_i + l + 2} \le R_i \le \frac{C_i^{\text{num}} + l}{C_i + l}
\end{equation}
The upper and lower bounds tend to~1 as $l\rightarrow\infty$, and the result follows from the sandwich principle.

\item
In this case, $g^{(C)}(n) = n^{\gamma^C}$, so
\begin{align*}
\min\{C \in \mathbb{N}\cup\{0\} \mid g^{(C)}(n) < 2\}
  & = \min\{C \in \mathbb{N}\cup\{0\} \mid \gamma^C \log n < \log(2) \} \\
  & = \min\{C \in \mathbb{N}\cup\{0\} \mid (1/\gamma)^C > \log_2 n) \} \\
  & = \min\{C \in \mathbb{N}\cup\{0\} \mid C > \log_{1/\gamma} (\log_2 n) \} \\
  & \le \lfloor \log_{1/\gamma} \log_2 n \rfloor + 1.
\end{align*}

\item
This follows from the fact that if $g(n) = a$, the resulting tree is a caterpillar with trees coming off the main ``spine" of the caterpillar having $n$ leaves. Within these trees the bound from $g$ becomes irrelevant, so the highest complexity among them is the fully balanced tree with complexity $\lceil \log_2(a) + 1\rceil$.
\end{enumerate}
\end{proof}

This matches the lower bounds for both complete binary trees, which are $\Omega(n\log n)$ and caterpillars (sticks with a single leaf on each internal node, $g(n) = 1$) which are $\Omega(n)$.
The suggests the intriguing possibility that Hein's algorithm is within a constant factor of instance-optimal with respect to the number of queries.
Testing that hypothesis would require investigating lower bounds for classes of unbalanced trees.

The class of $g$-beanstalks can be extended to allow a slow-growing fraction of the depths to violate the imbalance rule.
If $C_g(n)$ is the complexity bound for $g$-beanstalks, then allowing up to $O(C_g(n))$ layers to have balanced subtrees does not change the asymptotic growth of the complexity.

This raises the question of whether the $O(n k \log_k n)$ algorithm of~\cite{brodal_complexity_2001} may also be $o(n k\log_k n)$ on a narrower class of tree.
This would be surprising, since that algorithm uses less informative queries, which take three leaves and returns one of four tree representing their connectivity in $T^*$.
Indeed, the Theorem~\ref{thrm:Brodal_nlogn} below shows that it can improve by a factor of at most $k$, which is constant for the binary trees considered in this subsection.

The algorithm again inserts leaves  at a time while maintaining a tree $T$ on those leaves such that each edge in $T$ represents a path in $T^*$.
The algorithm achieves its efficiency by maintaining a balanced \emph{separator tree} on $T$.
A separator tree is a tree over the same set of nodes, though different edges, that maintains information about the subtrees of the tree.
Leaves of the original tree $T$ are also leaves of the separator tree.

\begin{theorem}\label{thrm:Brodal_nlogn}
     For every tree with $n$ leaves, there exists an order of presentation of the input such that the algorithm of~\cite{brodal_complexity_2001} has complexity $\Omega(n \log_k n)$.
\end{theorem}
\begin{proof}
    The separator tree for iteration $i$ has depth $\Omega(\log_k i)$.
    The proof will show that there exists a sequence in which to insert leaves such that separator tree will be searched to depth $\Omega(\log_k i)$ at each stage $i$, which shows that the whole algorithm has complexity $\Omega (\sum_{i=1}^n \log_k (i)) = \Omega (n \log_k n)$.
    We say that a leaf $z$ is inserted \emph{directly above} an existing leaf $y$ if it is added either as a child of $y$ or it is added by inserting a new internal node on the edge incident to $y$ and attaching $z$ to that new node.

    To ensure the separator tree always has to be searched to depth $\Omega(\log_k i)$, leaves can be inserted in such a way that each insertion is directly above a leaf -- either a child of the leaf's parent, or breaking the edge connected to the leaf.
    Separator trees are guaranteed to have the leaves of the tree at the bottom, so the edges directly above the leaves will also be at the bottom  of the separator tree.
    
    The first pair of leaves input to the algorithm will be such that the path between them includes $uv$, the edge root of the tree $T^*$ being inferred.
    Leaves will be added in an order that ensures $uv$ remains the edge root of the current estimate $T$ at each stage.

    Each subtree $T_u$,  $T_v$ is traversed in breadth first order, with interleaving in such an order as to ensure that $uv$ remains the edge root. ?
    When an internal node $r$ is reached, for each child $C$ of $r$, if no leaf in the subtree rooted at $C$ has been added, add the least deep leaf in that subtree, with ties broken arbitrarily.  
    All leaves will eventually be added, since the traversal continues to the leaves.
    
    All insertions will be to the directly above a leaf for the following reason.
    It is not possible that a node joins as a child of an ancestor of a branch point, since the leaves are added shallowest first.
    Let the latest common ancestor of newly inserted leaf $x$ of $T_i$ and any other inserted leaf $y$ of $T_i$ by $z$.
    All internal nodes that are ancestors of $z$ have already been discovered by the breadth-first traversal
    and so it is not possible for a new branch point to be above an existing branch point.
\end{proof}
The difference in lower bound may be due to the fact that the algorithm of~\cite{hein_optimal_1989} uses distance queries, which become more informative as the size of the tree grows, whereas that of~\cite{brodal_complexity_2001} uses less informative queries that return one of three distinct values, regardless of the size of the tree.

\section{Conclusion}
This paper has shown that Hein's algorithm was the first to determine the topology of an unrooted tree with $O(n k \log_k n)$ queries to an oracle that can only return the sum of edge weights between pairs of leaves.
Moreover, it has shown that the same algorithm needs at most $o(n k \log k n)$ queries for sufficiently unbalanced trees; this is possibly the first such result beyond algorithms tailored for the caterpillar.

This raises the interesting question of whether Heim's algorithm is within $O(1)$ of instance optimal.
That would hold if there is a constant $\eta \ge 1$ such that for all sufficiently large trees, $T^*$, the number of queries used by Heim's algorithm is at most $\eta$ times the number of queries used by the best algorithm to identify which leaf of $T^*$ is which.

this paper has only considered query complexity.
Efficient implementation is an open question.
Although the number of queries required is small, the cost of searching within the partial tree for the next place to join is asymptotically more.  A simple algorithm would take $O(n D)$ where $D$ is the diameter of the graph, which is not $O(n \log n)$ in general.  It is likely that an implementation with $O(n \log n)$ computational complexity exists for fixed $k$ (details forthcoming), but there is no indication that one with $O(n C(n))$ computational complexity does.  If so, that would imply the interesting result that the optimal query complexity is less than the inherent computational complexity.  
Indeed, the topologies such as the beanstalks for which Heim's algorithm is $o(n k \log_k n)$ benefit are also $\omega(\log_k n)$ in depth, suggesting a possible trade-off.
However, it is possible that a more sophisticated algorithm could obtain complexity $O(n(\log n + C(n)))$.

\newcommand{\etalchar}[1]{$^{#1}$}

\end{document}